\documentclass[a4paper]{article}
\pdfoutput=1

\usepackage{amsthm}
\usepackage{amsmath}
\usepackage{amssymb}
\usepackage{amsfonts}
\usepackage{epsfig}
\usepackage{graphics}
\usepackage{graphicx}
\usepackage{booktabs}
\usepackage{pdfpages}

\usepackage{complexity}
\usepackage{url}

\usepackage{latexsym}
\usepackage{psfrag}
\usepackage{enumerate}
\usepackage{here}
\usepackage{esvect}
\usepackage[update,prepend]{epstopdf}
\usepackage{anyfontsize}
\usepackage{url}
\usepackage{hyperref}
\usepackage{fullpage}
\usepackage[displaymath, left, mathlines, running]{lineno}

\usepackage{microtype}%if unwanted, comment out or use option "draft"

\newtheorem{theorem}{Theorem}

\newcommand{\ie}{{i.e.}}
\newcommand{\eg}{{e.g.}}

\newcommand{\alg}{\textsf{ALG}}
\newcommand{\opt}{\textsf{OPT}}

\newcommand{\area}{{\rm area}}

\newcommand{\RR}{\mathbb{R}} %  set of real numbers

\def\R{\mathcal R}

\def\Prob{{\rm Prob}}

\providecommand{\intd}[0]%
{\;\mbox{d}}

\newcommand{\e}{\mathrm{e}}

\usepackage[noline,linesnumbered]{algorithm2e}
\SetArgSty{textrm}
\let\oldnl\nl
\newcommand{\nonl}{\renewcommand{\nl}{\let\nl\oldnl}}
\makeatletter
\def\TitleOfAlgo{\@ifnextchar({\@TitleOfAlgoAndComment}{\@TitleOfAlgoNoComment}}
\def\@TitleOfAlgoAndComment(#1)#2{\nonl\hspace*{-1.5em}#2 #1\;}
\def\@TitleOfAlgoNoComment#1{\nonl\hspace*{-1.5em}#1\;}
\makeatother
\DontPrintSemicolon

\newcommand{\later}[1]{}
\newcommand{\old}[1]{}

\title{A Couple of Simple Algorithms for $k$-Dispersion}

\author{
  Ke Chen\thanks{%
    Department of Computer Science and Engineering, The Pennsylvania State University, PA, USA.
    Email~\texttt{kxc5915@psu.edu}}
  \and
  Adrian Dumitrescu\thanks{%
    Algoresearch L.L.C., Milwaukee, WI, USA.
    Email~\texttt{ad.dumitrescu@algoresearch.org}}
}

\begin{document}
%\linenumbers

\maketitle

\begin{abstract}
  Given a set $P$ of $n$ points in $\RR^d$, and a positive integer $k \leq n$, the $k$-dispersion problem
  is that of selecting $k$ of the given points so that the minimum inter-point distance among them
  is maximized (under Euclidean distances).   Among others, we show the following:

  \medskip
  (I) Given a set $P$ of $n$ points in the plane, and a positive integer $k \geq 2$, the $k$-dispersion problem
  can be solved by an algorithm running in $O\left(n^{k-1} \log{n}\right)$ time.
  This extends an earlier result for $k=3$,
  due to Horiyama, Nakano, Saitoh, Suetsugu, Suzuki, Uehara, Uno, and Wasa (2021) to arbitrary $k$.
  In particular, it improves on previous running times for small $k$. 

  \medskip
  (II) Given a set $P$ of $n$ points in $\RR^3$, and a positive integer $k \geq 2$, the $k$-dispersion problem
  can be solved by an algorithm running in
\[ \begin{cases}
    O\left(n^{k-1} \log{n}\right) \text{time}, & \text{if } k \text{ is even};\\
    O\left(n^{k-1} \log^2{n}\right) \text{time}, & \text{if } k \text{ is odd}.
  \end{cases} \]
For $k \geq 4$, no combinatorial algorithm running in $o(n^k)$ time was known for this problem.

  \medskip
  (III) Let $P$ be a set of $n$ random points uniformly distributed in $[0,1]^2$.
  Then under suitable conditions, a $0.99$-approximation for $k$-dispersion can be computed in $O(n)$ time
  with high probability. % at least $1 -1/\poly(n)$ when $n$ is large. 

\end{abstract}

\section{Introduction}\label{sec:intro}

The general \emph{dispersion} problem arises in selecting facilities to \emph{maximize} some function of the
distances between the facilities, \eg, the minimum or the average distance.
A geometric version was studied by Wang and Kuo~\cite{WK88}:
Given a set of $n$ location points in $\RR^d$ ($d$ is fixed) at which facilities may be placed and
a positive integer $k \leq n$,
the $k$-\emph{dispersion} problem (sometimes known also as the max-min $k$-dispersion problem) in $\RR^d$
is that of placing $k$ facilities so that the minimum pairwise distance between them is maximized.
Observe that since $k \leq n$, the optimization is over a nonempty set and so the problem is well defined,
and in general, a maximizing $k$-subset is not unique. 

Wang and Kuo~\cite{WK88} gave a polynomial algorithm for $d=1$ and proved that the problem is $\NP$-hard
already for $d=2$. The case $k=2$ in $\RR^2$ corresponds to the problem of computing the \emph{diameter}
of a planar point set and admits an optimal algorithm running in $O(n \log{n})$ time~\cite{PS85}. 
The case $k=3$ in $\RR^2$ admits an  algorithm running in $O(n^2 \log{n})$ time~\cite{HNS+21}.
Here we extend the above results for arbitrary $k \geq 2$.

\begin{theorem}\label{thm:plane}
  Given a set $P$ of $n$ points in the plane, and a positive integer $k \geq 2$, the $k$-dispersion problem
  can be solved by a combinatorial algorithm running in $O\left(n^{k-1} \log{n}\right)$ time.
\end{theorem}

All the algorithms mentioned above are combinatorial and so the algorithm in Theorem~\ref{thm:plane}
runs in $o(n^k)$ time for $k \geq 4$.

More generally, the $k$-dispersion problem in a complete graph $G$ with positive edge-weights, is that of selecting a
set of $k$ vertices in which every pair is connected by an edge of weight $\geq r$, such that
$r$ is maximized. 
For general $k$, Akagi, Araki, Horiyama, Nakano, Okamoto, Otachi, Saitoh, Uehara, Uno, and Wasa~\cite{AAH+18}
showed that the $k$-dispersion problem in graphs can be solved in
$O(n^{\omega \lfloor k/3 \rfloor + (k \pmod 3)} \log{n})$ time\footnote{The result of~\cite{NP85} for
  the $k$-clique problem is actually cited incorrectly as $O(n^{\omega k/3})$  instead
  $O(n^{\omega \lfloor k/3 \rfloor + (k \pmod 3)})$.}
using a Fast Matrix Multiplication (FMM) based algorithm for the $k$-clique problem
due to Ne{\v{s}}et{\v{r}}il and Poljak~\cite{NP85}.  
Here $\omega$ is the exponent of matrix multiplication~\cite[Ch.~10]{Mat10},
and $\omega < 2.372$ is the best known bound~\cite{DWZ23,VXXZ24}.

Using results from rectangular matrix multiplication due to Eisenbrand and Grandoni~\cite{EG04},
here we sharpen the above results for certain values of $k$.
Our improvements are listed against the running times from~\cite{AAH+18} (corrected as explained above)
in Table~\ref{table:k}. Note that the running times for our combinatorial algorithms when $k=4$ and $k=5$
beat those of the previous FMM based algorithms from~\cite{AAH+18}.

\begin{table}[h]
\centering
\scalebox{0.83}{
    \begin{tabular}{|c|c c c c c c c|} \hline
    $k$ & 2 & 3 & 4 & 5 & 6 & 7 & 8 \\[0.5ex] \hline
      Previous \cite{AAH+18} & $O(n \log{n})$ & $O(n^2 \log{n})$ & $O(n^{3.372})$ & $O(n^{4.372})$ &
      $O(n^{4.744})$ & $O(n^{5.744})$ & $O(n^{6.744})$ \\
    New (FMM \cite{EG04}) & &  & $O(n^{3.251})$ & $O(n^{4.086})$ & & $O(n^{5.590})$ & $O(n^{6.502})$ \\
    New (comb) & &  & $O(n^3 \log{n})$ & $O(n^4 \log{n})$ & $O(n^5 \log{n})$ & $O(n^6 \log{n})$ & $O(n^7 \log{n})$ \\\hline
\end{tabular}}
\vspace{0.1in}
\caption{Running times for $d=2$ and $k=2,\ldots,8$. For the exponents that are rounded up
  (for the FMM-based algorithms), the logarithmic factors are unnecessary.}
\label{table:k}
\end{table}

\begin{theorem}\label{thm:fmm}
  Given a set $P$ of $n$ points in the plane, and a positive integer $2 \leq k \leq 8$, the $k$-dispersion problem
  can be solved by a FMM based algorithm whose running time is specified in the second row of Table~\ref{table:k}.
  The same running times hold for $k$-dispersion in a complete graph on $n$ vertices with positive edge weights. 
\end{theorem}

An important advance for $k$-dispersion problem in the plane, particularly from the theoretical standpoint,
relies on the development of exact algorithm for the so-called \emph{parameterized independent set}
problem in unit disk graphs: given a set of $n$ unit disks in the plane, and a positive integer $k \leq n$,
find a set of $k$ non-intersecting disks, or report that no such set exists.
It is worth noting that $2x$ is a feasible solution to  $k$-dispersion if and only if
there exist $k$ disks of radius $x$ no pair of which intersect in their interior. 
The papers of Lev-Tov and Peleg~\cite{LP02}, Alber and Fiala~\cite{AF03}, Marx and Pilipczuk~\cite{MP15},
lead to a solution of the problem in $n^{O(\sqrt{k})}$ time based on \emph{geometric separators}.
Marx and Sidiropoulos~\cite{MS14} further generalized this result to higher dimensions. 
However, in some sense similar to FMM base algorithms, these algorithms are considered impractical due to
the large constants involved.
In particular, from a practical standpoint, it is unclear for what $k$ do these algorithms win in execution
speed over the combinatorial algorithms in Theorem~\ref{thm:plane}. 

We next consider algorithms for $k$-dispersion in $\RR^3$. 
As for the plane, the case $k=2$ corresponds to the problem of computing the \emph{diameter} of a point set and
admits an optimal algorithm running in $O(n \log{n})$ time~\cite{Ra01}. 
The case $k=3$ admits an  algorithm running in $O(n^2 \log^2{n})$ time~\cite{HNS+21}.
Here we extend the above results for arbitrary $k \geq 2$.
For $k \geq 4$, no combinatorial algorithm running in $o(n^k)$ time could be found in the literature,
although the dynamic programming machinery and higher dimensional generalization~\cite{MS14}
mentioned above suggest an $n^{O(k^{2/3})}$ time algorithm for $k$-dispersion in $\RR^3$.

\begin{theorem}\label{thm:3-space}
  Given a set $P$ of $n$ points in $\RR^3$, and a positive integer $k \geq 2$, the $k$-dispersion problem
  can be solved by a combinatorial algorithm running in
\[
  \begin{cases}
    O\left(n^{k-1} \log{n}\right) \text{time}, & \text{if } k \text{ is even};\\
    O\left(n^{k-1} \log^2{n}\right) \text{time}, & \text{if } k \text{ is odd}.
  \end{cases}
\]
\end{theorem}

The running times of the FMM-based algorithms for $k$-dispersion in $\RR^d$ are the same
as those for $k$-dispersion in $\RR^2$; refer to Table~\ref{table:k}.

\medskip
We next discuss $k$-dispersion from the standpoint of approximation.
Ravi, Rosenkrantz, and Tayi~\cite{RRT94} showed that in a general weighted graph setting
(\ie, edge weights are not required to satisfy the triangle inequality), there is
no polynomial time relative approximation unless $\P=\NP$. On the other hand,
the same authors~\cite{RRT94} showed that there is a ratio $1/2$ approximation algorithm
for $k$-dispersion in graphs that satisfy the triangle inequality.
Next we provide an efficient algorithm with a much better approximation for the case of uniformly distributed
random points in the unit square.

\begin{theorem}\label{thm:random}
Let $P$ be a set of $n$ random points uniformly distributed in $[0,1]^2$ and $3 \cdot 10^5 \leq k \leq n/(10^5 \ln{n})$.
Then a $0.99$-approximation for $k$-dispersion can be computed in $O(n)$ time with probability at least $1 -1/n$,
when $n$ is large. 
\end{theorem}

As in~\cite{WS11}, here we follow the convention that the approximation ratio of an algorithm
for a maximization (resp., minimization) problem is less than $1$ (resp., larger than $1$). 
%Throughout this paper all logarithms are in base~$2$.

\subparagraph{Related work.}
Computing order statistics in Euclidean space has been studied since the early days of Computational Geometry,
see the works of Chazelle~\cite{Chaz85},  Agarwal,  Aronov, Sharir, and Suri~\cite{AASS90},
Dickerson and Drysdale~\cite{DD91}, and more recently by Chan~\cite{Chan01}.
In particular, these works include algorithms for selecting the $k$th smallest distance among $n$ points in $\RR^d$.

A general taxonomy of dispersion problems in a graph setting as well as approximation algorithms
for many of these problems were proposed by Chandra and Halld{\'o}rsson~\cite{CH01}.

Ravi, Rosenkrantz, and Tayi~\cite{RRT94} gave an algorithm running in $O(n \log{n} + kn)$ for $k$-dispersion
with $n$ points on the line. Araki and Nakano~\cite{AN22} gave an algorithm running in $O(2^k k^{2k} n)$ time,
thereby providing a linear-time solution for fixed $k$ in this setting.

Consider the $1/2$ approximation due to Ravi, Rosenkrantz, and Tayi~\cite{RRT94}
for points in $\RR^d$. It is a greedy algorithms that repeatedly chooses points, one by one from the given $n$,
and adds them to an initially empty subset, so as to maximize the distance between the new point and those already selected.
A diameter pair is chosen in the first step\footnote{This choice appears not critical, since invariant (3) in
  their proof of Theorem~2 would still hold without it. Either way, taking a diameter pair in the first step is
a valid option for the greedy choice.}.
The algorithm terminates when the growing subset of selected points reaches size $k$. 
Interestingly enough, the same algorithm computes a $2$-approximation for the Euclidean $k$-\emph{center} problem:
given a set of $n$ points, here the goal is to select $k$ of them so that the maximum distance
from any point in the set to its closest point in the subset is \emph{minimized}; see~\cite[Chap.~4.2]{HP11},
and~\cite[Chap.~2.2]{WS11}.

Another research direction in connection with dispersion is as follows.
Let $\R$ be a family of $n$ subsets of a metric space.
The problem of {\em dispersion in $\R$} is that of selecting $n$ points,
one in each subset, such that the minimum inter-point distance is maximized.
This dispersion problem was introduced by Fiala et~al.~\cite{FKP05}
as \emph{systems of distant representatives},
generalizing the classic problem \emph{systems of distinct representatives}.
Fiala et~al.~\cite{FKP05} showed that dispersion in unit disks is already NP-hard.
Approximation algorithms with small constant ratios for the case when 
$\R$ is a set of unit disks in the plane were first obtained by Cabello~\cite{Ca07},
and later improved and extended to arbitrary disks by Dumitrescu and Jiang~\cite{DJ12}.
In another direction, given a finite family of convex bodies in $\RR^d$,
the same authors~\cite{DJ15} gave a sufficient condition for the existence of a system
of distinct representatives for the objects that are also distant from each other.

\subparagraph{Preliminaries.}
Recall that $\omega < 2.372$ is the exponent of matrix multiplication~\cite{Mat10},
namely the infimum of numbers $\tau$ such that two $n \times n$ real matrices can be multiplied in
$O(n^\tau)$ time (operations). Similarly, let $\omega(p,q,r)$ stand for the infimum of numbers $\tau$ such that
an $n^p \times n^q$ matrix can be multiplied by an $n^q \times n^r$ matrix in $O(n^{\tau})$ time (operations). 
Throughout this paper, $\log{x}$ and $\ln{x}$ denote the logarithms of $x$ in base $2$ and $\e$, respectively.

\section{Combinatorial algorithms}  \label{sec:comb}

\subparagraph{$k$-dispersion in the plane.} 
We give a recursive algorithm for the planar version; the same algorithm, with an updated
analysis, however, can be used to solve the $k$-dispersion problem in $\RR^d$.
The algorithm includes a closest pair of points of the returned $k$-subset in the output.

\begin{proof}[Proof of Theorem \ref{thm:plane}]
  We prove the theorem by induction on $k$. The cases $k=2$ and $k=3$ have been already verified
  (in Section~\ref{sec:intro}) and they provide the basis for the induction. We now prove the statement
  for $k \geq 4$, assuming that it holds for previous values.

  Note that the minimum distance, say, $x$, in an optimal $k$-subset $K \subset P$ of points
  corresponds to a closest pair of points in $K$. The algorithm correctly guesses $x$ by scanning all
  ${n \choose 2}$ pairs of points in $P$.

  \newpage
%\medskip
\noindent{{\bf Algorithm~\texttt{A}$(k)$}}
\begin{itemize} \itemsep 1pt
\item[] Set $\mathit{current\_best}=0$ and $\mathit{result}=\emptyset$;
\item[] For each pair of points $a,b \in P$:
\item []
  \begin{itemize} \itemsep 1pt
    \item [] Step 1. Set $x:= |ab|$;
    \item [] Step 2. Remove from $P$ every point $p$ such that $|pa| < x$ or $|pb| < x$;
    \item [] Step 3. Let $P'$ denote the remaining set;
    \item [] Step 4. If $|P'| < k-2$, skip to the next pair $a,b$, else
      run Algorithm~\texttt{A}$(k-2)$ on $P'$ and let $K'$ be the output $(k-2)$-subset; 
    \item [] Step 5. If the minimum distance in $K'$ is less than $x$,  skip to the next pair $a,b$;
    \item [] Step 6. If $x>\mathit{current\_best}$, set $\mathit{current\_best}=x$ and
      set $\mathit{result}$ to be $\{a, b\}\cup K'$ with $a,b$ as a closest pair;
      %else return $\{a,b\} \cup K'$, and $a,b$ as a closest pair.
  \end{itemize}
\item[] Return $\mathit{result}$.
\end{itemize}

It is clear that the algorithm works correctly, since by induction Algorithm~\texttt{A}$(k-2)$
returns a $(k-2)$-subset of $P'$ with the largest minimum distance.
Note that Step 2 takes $O(n)$ time
and that Algorithm~\texttt{A}$(k-2)$ runs on a set of at most $n$ points, which by the induction hypothesis,
takes $O\left(n^{k-3} \log{n}\right)$ time, for $k \geq 4$.
Thus the running time is at most
\[ {n \choose 2} \cdot O \left(n + n^{k-3} \log{n} \right)  = O\left(n^{k-1} \log{n}\right), \]
completing the induction step and thereby the proof of the theorem.
\end{proof}

\subparagraph{$k$-dispersion in $d$-space.} 

\begin{proof}[Proof of Theorem \ref{thm:3-space}]
 Recall that $d=3$. We use the same recursive Algorithm~\texttt{A}$(k)$.
  The only difference is in the basis of the recursion, where the solutions for $k=2$ and $k=3$, take
  $O(n \log{n})$ and  $O(n^2 \log^2{n})$ time, respectively; see~\cite[Thm.~8]{HNS+21} for the case $k=3$.
\end{proof}

\subparagraph{Remark.} Taking into account that the diameter of $n$ points in $\RR^d$ can be computed
by an algorithm of Yao~\cite{Yao82} in time\footnote{These running times are misprinted in~\cite[Thm.~8]{HNS+21}.}
\[ O\left(n^{2 - \alpha(d)} (\log{n})^{1 - \alpha(d)} \right), \text{ where } \alpha(d) = 2^{-(d+1)}, \]
the same analysis of Algorithm~\texttt{A}$(k)$ yields that the $k$-dispersion of $n$ points in $\RR^d$
can be computed in
\[
  \begin{cases}
   O\left(n^{k - \alpha(d)} (\log{n})^{1 - \alpha(d)} \right) \text{time}, & \text{if } k \text{ is even};\\
   O\left(n^{k - \alpha(d)} (\log{n})^{2 - \alpha(d)} \right) \text{time}, & \text{if } k \text{ is odd}.
  \end{cases}
\]

Indeed, as shown in~\cite[Thm.~8]{HNS+21}, the $k$-dispersion of $n$ points in $\RR^d$
can be computed in time
\[ O\left(n^{2 - \alpha(d)} (\log{n})^{1 - \alpha(d)} \right), \text{ and }
   O\left(n^{3 - \alpha(d)} (\log{n})^{2 - \alpha(d)} \right), \]
for $k=2$ and $k=3$, respectively, verifying the induction base.
Note that all these algorithms for $k$-dispersion run in $o(n^k)$ time.

\section{FMM based algorithms for $k$-dispersion in graphs}  \label{sec:FMM}

We next discuss algorithms for $k$-dispersion in graphs that use Fast Matrix Multiplication (FMM).
It is worth noting that these algorithms \emph{can} be used for $k$-dispersion in $\RR^d$.

\begin{proof}[Proof of Theorem \ref{thm:fmm}]
The improved running times specified in the second row of Table~\ref{table:k} follow from results
on fast rectangular multiplication due to Eisenbrand and Grandoni.
For each $k$, the algorithm relies on a specific rectangular multiplication combination.
Specific details on these instantiations can be found in~\cite[pp.~12--13]{CDL25};
entries from~\cite[Table~3]{LU18} and~\cite[Table~1]{VXXZ24} are relevant; see also~\cite{DL23}.
Let $e(k)$ denote the exponent that appears in the running time, \ie, the algorithm for $k$-dispersion
runs in $O(n^{e(k)})$ time; these are the exponents that appear in the running times in row $2$ of Table~\ref{table:k}.
Recall that $\omega< 2.372$.

\begin{itemize} \itemsep 2pt

\item $e(4) = \omega(1,1,2) = \omega(1,2,1) < 3.251$; see~\cite{EG04}.

\item $e(5) =\omega(2,1,2) = 2 \omega(1,0.5,1) < 2 \cdot 2.043 = 4.086$;  see~\cite{EG04,VXXZ24}.

\item $e(6)  = 2 \omega < 4.744$; see~\cite{IR78,NP85}.

\item $e(7) = \leq \omega(2,3,2) = 2 \omega(1,1.5,1) < 2 \cdot 2.795 = 5.590$; see~\cite{EG04}.

\item $e(8) = 2 \omega(1,1,2) = 2 \omega(1,2,1) < 2 \cdot 3.251 = 6.502$; see~\cite{EG04}. \qedhere
  
\end{itemize}
\end{proof}

\section{Approximation algorithm for random points}   \label{sec:random}

We start with the description of the algorithm.

\medskip     
\noindent{{\bf Algorithm~\texttt{B}$(k)$}} %\linebreak
Input: $n$ points randomly and uniformly distributed in $U=[0,1]^2$
\begin{itemize} \itemsep 2pt
\item  [] Step 1. Lay out (in an arbitrary fashion) a triangular lattice of side length
  $\sqrt{3} \, y$, where $y = \sqrt{0.995} \cdot \left(\frac{4}{27} \right)^{1/4} \cdot \frac{1}{\sqrt{k}}$;
\item  [] Step 2. Let $U_1 = [y/2, 1-y/2]^2$, $\Lambda_1$ be the set of lattice points in $U_1$,
  and $\Omega_1$ be the set of (small) disks of radius $r=y/240$ centered at points in $\Lambda_1$;
\item  [] Step 3. Distribute the $n$ points one by one to the appropriate disk in $\Omega_1$ or to a leftover list;
\item  [] Step 4. Arbitrarily select a point in each disk, if there is, and append it to an initially empty
  output list. Stop when the output list has reached size $k$;
\end{itemize}

\begin{proof}[Proof of Theorem \ref{thm:random}]
We first deduce  a lower bound on $\opt$, the distance in an optimal solution. 
Let $K$ be an optimal solution with $k$ points for a given instance, 
$\opt=2x$ be the minimum inter-point distance in $K$, $|K|=k$, and $U_2 = [-x, 1+x]^2$. 
Observe that the disks of radius $x$ centered at the points in $K$ are pairwise disjoint
and contained in $U_2$.  We have $\area(U_2) = (1+2x)^2$. 

The square $U_2$ is a so-called \emph{tiling domain}, \ie, a domain that can be used to
tile the whole plane~\cite[Ch.~3.4]{FFK23}. Recall that $k \geq 3 \cdot 10^5$. 
A packing argument (as in~\cite[p.~66]{FFK23}) requires that
\[ k \pi x^2 \leq \frac{\pi}{\sqrt{12}} \cdot \area(U_2), \text{ or } 
%\sqrt{12} k x^2 \leq  (1+2x)^2, \text{ or } 
x \leq \frac{1}{12^{1/4} \sqrt{k} -2} \leq \frac{1.002}{12^{1/4} \sqrt{k}}, \]
where the second inequality follows from the assumption on $k$. As such, we have
\begin{equation} \label{eq:opt} 
\opt=2x \leq  \frac{2 \cdot 1.002}{12^{1/4} \sqrt{k}}. 
\end{equation}

\medskip
Next we analyze the algorithm and the quality of the solution produced by it.
Since every input point can be assigned to the appropriate small disk in $\Omega_1$, if any,
in $O(1)$ time, it is clear that the algorithm takes $O(n)$ time. 
By the assumption on $k$, we have
\[ y = \sqrt{0.995} \cdot \left(\frac{4}{27} \right)^{1/4} \cdot \frac{1}{\sqrt{k}} \leq 0.002. \]
The square $U_1$ is clearly a tiling domain. We have $\area(U_1) \geq (1-y)^2 \geq 0.998^2 \geq 0.995$. 
By construction, the disks of radius $y$ centered at the points in $\Lambda_1$ cover $U_1$.
Finally, observe that all disks in $\Omega_1$ are entirely contained in $U$. 

Let $m=|\Lambda_1|$. A packing argument (as in~\cite[p.~66]{FFK23}) requires that
\[ m \pi y^2 \geq \frac{2 \pi}{\sqrt{27}} \cdot \area(U_1), \text{ or }
m \geq \frac{2}{\sqrt{27}} \cdot \frac{(1-y)^2}{y^2} \geq
\frac{2}{\sqrt{27}} \cdot  \frac{0.995}{0.995} \cdot \frac{\sqrt{27}}{2} \cdot k = k. \]
It follows that Algorithm~\texttt{B} outputs $k$ points whose pairwise distances
are at least
\[ \sqrt{3} y - 2r = \sqrt{3} y \left(1 - \frac{2r}{\sqrt{3} y}\right)  \geq 0.995 \sqrt{3} y =
0.995^{3/2} \cdot \sqrt{3} \cdot \left(\frac{4}{27} \right)^{1/4} \cdot \frac{1}{\sqrt{k}}. \]
That is, the minimum inter-point distance, $\alg$, in the solution constructed by the algorithm satisfies
\begin{equation} \label{eq:alg} 
  \alg \geq 0.995^{3/2} \cdot \sqrt3 \cdot \left(\frac{4}{27} \right)^{1/4} \cdot \frac{1}{\sqrt{k}}.
\end{equation}

The resulting approximation ratio of the algorithm is 
\begin{align*}
 \frac{\alg}{\opt} &\geq 0.995^{3/2} \cdot \sqrt{3} \cdot \left(\frac{4}{27} \right)^{1/4} \frac{12^{1/4}}{2 \cdot 1.002}
  = \sqrt{3} \cdot \frac{0.995^{3/2}}{1.002} \cdot \frac{\sqrt{2} \cdot \sqrt{2} \cdot 3^{1/4}}{\sqrt{3} \cdot 3^{1/4} \cdot 2}\\
  &= 0.995^{3/2}/1.002 \geq 0.99. 
\end{align*}

It remains to show that with high probability, there exist at least $k$ nonempty disks. Fix any set
$\Omega'_1 \subseteq \Omega_1$ of $k$ disks (out of $m$). We show that the probability that at least one of them
is empty is small, \ie, at most $1/\poly(n)$. Thus with high probability each of them is nonempty, as desired.

Consider a fix disk $D \in \Omega'_1$. Recall that its radius is $r= \frac{1}{240}y$  and that $D$ is contained in $U$.
Let $E_D$ be the event that $D$ is empty of points in $P$, and let $E$ be the event that at least one disk
in $\Omega'_1$ is empty of points in $P$. We bound from below the area of each disk in $\Omega'_1$:
\begin{align*}
r &=\frac{y}{240} = \frac{1}{240} \cdot \sqrt{0.995} \cdot \left(\frac{4}{27} \right)^{1/4} \cdot \frac{1}{\sqrt{k}}
\geq \frac{1}{390 \sqrt{k}}, \text{ whence } \\
\pi r^2 &\geq \frac{1}{50000 k}.
\end{align*}
Since the $n$ points are randomly and uniformly distributed in $U$, we have
\[ \Prob(E_D) \leq (1-\pi r^2)^n \leq \left( 1 -\frac{1}{50000 k} \right)^n \leq \exp\left(\frac{-n}{50000 k}\right)
\leq \exp\left(-2\ln{n} \right) = \frac{1}{n^2}, \]
by applying the standard inequality $\left(1-x\right)^{1/x} \leq 1/\e$ for $0<x<1$. 
By the union bound~\cite[Lemma~1.2]{MU17}, it follows that
\[ \Prob(E) \leq k \cdot \Prob(E_D) \leq \frac{1}{n}, \]
as claimed. 
\end{proof}

\subparagraph{Remark.} It is clear that  Algorithm~\texttt{B} can be adjusted so that it works 
in a less constrained setting (\eg, also for a smaller $k$) at the cost of reducing the approximation ratio,
say, to $0.9$. For instance, the requirement in Theorem~\ref{thm:random} could be relaxed to
$900 \leq k \leq n/(2500\ln n)$ to achieve a $0.9$ approximation.
Likewise, it is also clear that the algorithm can be adjusted in the opposite direction to boost
its approximation ratio beyond $0.99$.

\section{Concluding remarks} \label{sec:remarks}

We highlight two questions of interest:

\begin{enumerate} \itemsep 2pt

\item Is there an approximation algorithm for $k$-dispersion in the plane with a constant ratio above $1/2$? 

\item What is an approximate switchover value for $k$, at which the parameterized independent set algorithms for $k$-dispersion
  in the plane running in $n^{O(\sqrt{k})}$ time, would be faster than the combinatorial algorithms in Theorem~\ref{thm:plane}?

\end{enumerate}

\end{document}